\documentclass[12pt,a4paper]{article}

\usepackage{epsf}
\usepackage{epsfig}
\usepackage{amsfonts}
\usepackage{amsmath}
\usepackage{amssymb}
\usepackage{amsthm}
\usepackage{setspace}
\usepackage{multirow}
\usepackage{graphics}
\usepackage{graphicx}
\usepackage{geometry}
\usepackage{array}
\usepackage{float}
\usepackage{lscape}
\usepackage{color}
\usepackage{calligra}
\usepackage{dsfont}
\usepackage{tabularx}
\usepackage{dcolumn}
\usepackage[bottom]{footmisc}
\usepackage[round,sort&compress]{natbib}
\usepackage{enumerate}
\usepackage{url}
\usepackage[normalem]{ulem}
\usepackage{epstopdf}
\usepackage[space]{grffile}
\newtheorem{proposition}{Proposition}
\newtheorem{theorem}{Theorem}

\newtheorem{Cor}{Corollary}

\DeclareMathAlphabet{\mathpzc}{OT1}{pzc}{m}{it}
\DeclareMathAlphabet{\mathcalligra}{T1}{calligra}{m}{n}
\begin{document}

\newgeometry{textheight=25.0cm, voffset = -20pt}
\begin{titlepage}
\singlespacing

    \title{Sophisticated Attacks on Decoy Ballots:\\ The  Devil's Menu and the Market for Lemons\thanks{We would like to thank David Basin, Afsoon Ebrahimi, Georgy Egorov, Lara Schmid, Salvador Barber\`{a} and seminar participants at ETH Zurich for valuable comments. All errors are ours.}}

	\author{Hans Gersbach \\
		\small CER-ETH -- Center of Economic \\
		\small Research at ETH Zurich and CEPR \\
		\small Z\"urichbergstrasse 18 \\
		\small 8092 Zurich, Switzerland \\
		\small hgersbach@ethz.ch \vspace{5mm}\\
	\and Akaki Mamageishvili \\
		\small CER-ETH -- Center of Economic \\
		\small Research at ETH Zurich  \\
		\small Z\"urichbergstrasse 18 \\
		\small 8092 Zurich, Switzerland \\
		\small amamageishvili@ethz.ch \vspace{5mm}\\
	\and Oriol Tejada \\
		\small CER-ETH -- Center of Economic \\
		\small Research at ETH Zurich  \\
		\small Z\"urichbergstrasse 18 \\
		\small 8092 Zurich, Switzerland \\
		\small toriol@ethz.ch \vspace{5mm}\\
	}
		\date{\normalsize First version: July 2017\\ \vspace{1mm}
			This version:  December 2017}
	\maketitle\thispagestyle{empty}

\vspace{-1cm}
\singlespacing

    \begin{abstract}
      Decoy ballots do not count  in election outcomes, but otherwise they are indistinguishable from real ballots. By means of a game-theoretical model, we show that decoy ballots may not provide effective protection against a malevolent adversary trying to buy real ballots. If the citizenry is divided into subgroups (or districts), the adversary can construct a so-called ``Devil's Menu'' consisting of several prices. In equilibrium, the adversary can buy the real ballots of any strict subset of districts at a price corresponding to the willingness to sell on the part of the citizens holding such ballots. By contrast, decoy voters are trapped into selling their ballots at a low, or even negligible, price. Blowing up the adversary's budget by introducing decoy ballots may thus turn out to be futile. The Devil's Menu can also be applied to the well-known ``Lemons Problem''.  

\medskip 


\medskip


 \end{abstract}

\end{titlepage}

\restoregeometry

\onehalfspacing

\setcounter{page}{2}

\section{Introduction}

In the past few years electronic voting has become popular in many countries.\footnote{See \url{https://en.wikipedia.org/wiki/Electronic_voting_by_country}, retrieved on 18 September, 2017.} The possibility that voting can be carried out electronically opens up myriad new options for representative and direct democracies alike,  as the marginal cost of voting will be typically lower than for normal elections. One possibility is to randomly select a subgroup of citizens from the entire population and have each of them vote on a single issue. Such a voting scheme is called \textit{random sample voting} \citep{random-sample-voting}. Assuming that the chosen subpopulation (or \textit{sample voting group}) can be trusted to represent the entire citizenry,  random sample voting may improve decision-making by yielding the same decision as standard voting, albeit at a less costly voter-participation level.\footnote{In addition, several issues could be considered at once by having different subpopulations vote separately on each of them. This way, citizens may acquire more information only about the issue on which they have a say, and decisions may be  better informed.}


Electronic voting procedures of this type entail risks. For instance, a malevolent third party---which we will henceforth call \textit{adversary}---may be interested in buying some citizens' right to vote in a particular instance of random sample voting.\footnote{Vote-buying is also a problem in standard elections, and one which has already been studied in the literature \cite[see e.g.][]{dekel2008vote,finan2012vote}.} To prevent this from happening, \cite{random-sample-voting} has proposed the following mechanism: not only the members of the sample voting group receive a ballot, but so do all other  citizens. The difference is that only the ballots of the members of the sample voting group are real, i.e., only they will be counted. The remaining ballots act as a decoy---and hence they are called \textit{decoy ballots}. Crucially, whether a ballot is real or a decoy is a citizen's private information,  so the adversary cannot distinguish real from decoy ballots.\footnote{One could conceive of a ballot as a password in the electronic voting system, where a decoy ballot is simply an invalid password.} Voting systems based on decoy ballots revolve around the idea that selling them to an adversary is valuable for society---and can  thus constitute a social norm---, since doing so may prevent adversaries from buying a large amount of real ballots and then manipulating the voting outcome.\footnote{Parallels can be drawn between decoy ballots preventing vote buying and the idea of producing fake drugs and selling them in the drug market to destroy this market.} \cite{thwarting-vote-buying} have recently shown that, under some assumptions, decoy ballots \textit{do} indeed discourage attempts by an adversary with a limited budget to try and buy (real) ballots.\footnote{For other types of attack on electronic voting, see \cite{basinelection}.} 

In this paper we show that when the whole population with a right to vote is divided into several subgroups or \textit{districts}, the adversary can employ an efficient procedure based on a particular price-offering scheme. With such a procedure, the adversary will be able to purchase all real ballots of a predetermined number of districts $q$ that is lower than the total number of districts, each at a price equal to the willingness to sell on the part of the citizens holding such ballots. Decoy ballots will also have to be bought by the adversary, yet at a lower, possibly even negligible, price. 

The procedure we suggest is called {\it Devil's Menu} and is presented in two forms: a \textit{weak form} and a \textit{strong form}. In both variants, the adversary will offer citizens two slots, called \textit{slot~1} and \textit{slot~2}.  Citizens will choose one of them at most---they are not obliged to participate. Each slot  in its turn is associated with two prices, and the final price that each citizen  is offered in exchange for his ballot (real or decoy) will depend on the slot  chosen and on the final \textit{status} of the district to which he belongs. In other words, the adversary offers each citizen a four-price scheme. The final status of a district will be \textit{selected} if the ratio of slot-$1$ applicants to citizens holding real ballots is among the set of the smallest $q$ ratios for all districts, with ties being broken by fair randomization. Additionally, prices associated with slot~$1$  will be higher in selected districts than in \textit{non-selected} districts. For the procedure to work, it is crucial that, once the corresponding districts have been selected, \textit{all} citizens  who have applied for either slot can decide whether they want to sell their vote at the prevailing price or not.  By contrast, the adversary must be committed to buying the ballots at the prices announced. 

 A Devil's Menu places citizens holding a decoy ballot  in the following dilemma: Applying for slot~$1$ is more attractive, as long as their district is selected. However, by doing so, these citizens  incur the risk that their district will in fact no longer be selected.  For non-selected districts, prices are very low in both slots,  albeit marginally higher in slot~$2$ than in slot~$1$, so voters holding a decoy ballot will choose the latter. In equilibrium, there  is a positive probability that their district  will be selected, so the price they expect to be offered in exchange for their ballot is higher than the sure payoff associated with slot~$1$ in a non-selected district. Citizens holding decoy ballots will nonetheless be paid lower prices for their ballots than voters holding a real ballot. 

We prove four specific results. First, by setting the four prices appropriately, the adversary will induce a unique equilibrium of the underlying game, in which citizens with a real ballot apply for slot~$1$ and citizens with a decoy ballot apply for slot~$2$. This is the weak form of the Devil's Menu. It guarantees that the price decoy voters will receive in this (unique) equilibrium will be significantly lower than the price at which the real voters will be willing to sell, though it will not be negligible. Second, running the Devil's Menu sequentially  and targeting all real votes of only one district each time reduces the budget required by the adversary by lowering the price eventually paid for decoy ballots. Third, this latter price can in fact be set arbitrarily low, thereby eliminating almost all superfluous expenditures. The adversary's budget can then be used almost entirely for real ballots,  so blowing up the adversary's budget  with decoy ballots may prove futile. This is the strong form of the Devil's Menu. There are two variants of the Devil's Menu that ensure this desirable circumstance: on the one hand, the price associated with slot~$2$ in non-selected districts can be made arbitrarily low---up to the point where only three different prices are  in fact offered---,  though this opens up the possibility that additional equilibria may exist; on the other, a more sophisticated six-price Devil's Menu may be used, which again ensures uniqueness of the equilibrium targeted. This latter variant requires that the adversary can offer two different prices to applicants of slot~$2$ that belong to selected districts, based on the districts' interim status. Fourth and last, if the adversary has strong commitment power, she can achieve this same outcome even if the citizenry is not divided into subgroups (i.e., there is only one district).

The remainder of the paper is organized as follows: In Section~\ref{sec:model} we present our model. In Section~\ref{sec:mechanism} we introduce the Devil's Menu with four prices and show its weak form. In Section~\ref{sec:costless_buying} we analyze alternative forms of the Devil's Menu, including two variants of its strong form.  Section~\ref{sec:extensions} discusses extensions to the model and applications of the Devil's Menu to the  well-known ``Lemons Problem.'' Section~\ref{sec:conclusion} concludes.

\section{The Model}\label{sec:model}

There is a finite set of risk-neutral citizens denoted by $N$. Each citizen $i\in N$ may be of one of two types: if $t_i=R$, he has a \textit{real ballot}; if $t_i=D$, he has a \textit{decoy ballot}. The former citizens are henceforth called \textit{real voters}, the latter \textit{decoy voters}. We assume that $N$ is partitioned into $\bar{k}>1$ subgroups or \textit{districts} $N_1,\ldots,N_{\bar{k}}$, with $n_k^R$ and $n^D_k$ denoting the number of ballots of each type within $N_k$, for all $k=1,\ldots,\bar{k}$. Natural partitions of the citizenry  may be based on states, cities, and villages. We assume that  while $(n_k^R,n_k^D)_{k=1}^{\bar{k}}$ is common knowledge, the type of each individual citizen is private information. For any set $S$, it will be convenient throughout to let $s$ denote its cardinality. Accordingly, $n=n_1+\ldots+n_{\bar{k}}$ and $n_k=n_k^R+n_k^D$, for all $k=1,\ldots,\bar{k}$. We consider that $n_1, n_2 ,\ldots , n_{\bar{k}}>1$ and let $n^D:=n^D_1+\ldots+n^D_{\bar{k}}$ and $n^R:=n-n^D$. 

For simplicity, we assume that any real voter will value his ballot at $V>0$, while any decoy voter will value his ballot at $0$.\footnote{This assumption is discussed in Section~\ref{subsec:general_valuation}.} Both assumptions are common knowledge. Besides the citizens, there is an \textit{adversary}, henceforth denoted by $A$. For the sake of clarity, throughout the paper we will use ``she'' to refer to the adversary and ``he'' to refer to any voter. The goal of the adversary is to buy all real ballots of exactly $q$ districts at minimal cost, with $1\leq q\leq \bar{k}- 1$.\footnote{If decoy ballots were applied in voting procedures with a structure similar to that of the US presidential elections, the adversary might only be interested in buying real ballots of swing states. We proceed on the assumption that identities of districts are not relevant for the adversary,  so the different procedures we suggest  would have to be adapted if that were not the case.} A first attempt based on a  brute-force attack would be to try and buy all ballots at a price $V+\varepsilon$, for some $\varepsilon>0$ arbitrarily small.   Doing so is feasible if $A$'s budget, which we denote by $B$, is large enough. It is sufficient to assume that
\begin{equation}
\label{eq:budget_condition}
B \geq\max_{S \subseteq  \{1,\ldots,\bar{k}\}, s=q} \left[ (V+\varepsilon) \cdot \sum_{k\in S} n_k \right],
\end{equation} 
and that this is common knowledge. The right-hand side in the above inequality is the amount the adversary would have to spend if she bought \textit{all} ballots from an arbitrary selection of $q$ districts. To do so, she would offer to pay $V+\varepsilon$ to all citizens in these districts in exchange for their ballots. Because the price offered to the citizens would be higher than their ballot valuation, regardless of the type of ballot they hold, they would \textit{all} accept the transaction. Yet it is clear that this  brute-force attack is very inefficient, since the adversary has to buy all ballots, real and decoy, at the same high price, namely $V+\varepsilon$. In the next section, we show that the adversary can  achieve the same goal of buying all real ballots of $q$ districts at a much lower cost if she employs what we call a Devil's Menu.




\section{The Devil's Menu}\label{sec:mechanism}

First, we describe the mechanism that guarantees that the adversary will be able to buy all real ballots in $q$ districts, each at price $V+\varepsilon$, without having to pay this same amount to all decoy voters in these districts. For its successful operation, the adversary must be committed to its rules---i.e. she  cannot change its functioning at any point in time---, and this property must be common knowledge among the citizens, who will then simply envision the mechanism as the execution of an algorithm. As a tie-breaking rule, we assume that the individual citizen will not offer his ballot to the adversary if he is indifferent between his valuation of the ballot and the price offered by the adversary. The mechanism works as follows:

\begin{enumerate}	
\item  Adversary $A$ offers two slots to all citizens: slot~one (denoted by $s_1$) and slot~two (denoted by $s_2$). Depending on the ratio of slot-$s_1$ applicants to real voters across districts, the adversary will select $q$ districts, as formulated in Step~3. Moreover, $A$ offers the following prices for ballots:

\begin{table}[H]
	\begin{center}
		\def\arraystretch{1.5}
		\begin{tabular}{ c|c|c| } 
			   & slot $s_1$ & slot $s_2$ \\  		\hline
			 district selected & $p_1^{se}=V+\varepsilon$ & $p_2^{se}=\delta$ \\ 	\hline
			district not selected & $p_1^{ns}=\varepsilon$& $p_2^{ns}=2\varepsilon$ \\ 	\hline
		\end{tabular}
	\end{center}
	\caption{The Devil's Menu with four prices---depending on the slot applied for (column) and the final status of  the voter's district (row).}
	\label{table_prices}
\end{table}

In this four-price scheme, $\varepsilon>0$ is assumed to be arbitrarily small, while $\delta$ satisfies $2\varepsilon<\delta\leq V-\varepsilon$.

\item Each citizen either applies for one of the two slots or abstains. For each $k=1,\ldots,\bar{k}$, we let $m_k$ be the number of citizens from $N_k$ who apply for slot~$s_1$, and we set $\rho_k:=\frac{m_k}{n^R_k}$.

\item  Assume w.l.o.g that $\rho_1 \leq \rho_2 \leq \ldots \leq \rho_{\bar{k}}$ and let $\rho : = \rho_q$. Then define the following three sets of districts:
\begin{align*}
C=\{N_k \vert k\in \{1,\ldots,\bar{k}\}, \rho_k<\rho\},\\
T=\{N_k \vert k\in \{1,\ldots,\bar{k}\}, \rho_k=\rho\},\\
O=\{N_k \vert k\in \{1,\ldots,\bar{k}\}, \rho_k>\rho\}.
\end{align*}
Accordingly, we have $\bar{k}=c+t+o$ and $t\geq \max \{1,q-c\}$. All districts in $C$ are \textit{selected} by $A$. From the set $T$, $q-c$ districts are chosen by fair randomization so that their \textit{final status} is also selected. That is, a district in $T$ has a probability $\frac{q-c}{t}$  of being selected by the mechanism---and then belonging to $T$ is its \textit{interim status}. The final status of the remaining districts is \textit{non-selected}.


\item Once all districts are assigned their final status---and  hence the price offered to each applicant has been determined---, these citizens decide whether to sell their vote or not at the prevailing price. The adversary is obliged to accept the transaction at the request of the citizens.

\end{enumerate}

The \textit{Devil's Menu} is a four-price scheme in which individuals sequentially commit first, to a subset of prices within a slot, and, second, to the possibility of selling their vote once the price has been determined for the slot they  have applied for. These final prices depend on the choices by \textit{all} the citizens. In turn, the adversary commits to the correct execution of the mechanism. It will be useful to display the (expected) payoff matrices of voters. Since a decoy voter will always sell his vote in Step 4, his expected payoff is given in Table~\ref{table_payoffs}. 

\begin{table}[H]
	\begin{center}
		\def\arraystretch{1.5}
		\begin{tabular}{ c|c|c|c| } 
			& $C$ & $T$ & $O$ \\  		\hline
			$s_1$ & $V+\varepsilon$ & $\frac{q-c}{t}\cdot (V+\varepsilon)+ \frac{t-q+c}{t}\cdot \varepsilon $ & $\varepsilon$ \\ 	\hline
			$s_2$ & $\delta$ & $\frac{q-c}{t}\cdot \delta  + \frac{t-q+c}{t}\cdot2\varepsilon$ & $2\varepsilon$ \\ 	\hline
		\end{tabular}
	\end{center}
	\caption{Expected payoff of a decoy voter in the Devil's Menu with four prices---depending on the slot  applied for (column) and the final status of  the voter's district (column).}
	\label{table_payoffs}
\end{table}

Real voters will, of course, not sell their vote if their district is not chosen, nor will they sell their vote if they  have applied for slot~$s_2$. Their expected payoff is then given in Table~\ref{table_expected_payoffs_realvoter}.

\begin{table}[H]
	\begin{center}
		\def\arraystretch{1.5}
		\begin{tabular}{ c|c|c|c| } 
			& $C$ & $T$ & $O$ \\  		\hline
			$s_1$ & $V+\varepsilon$ & $\frac{q-c}{t}\cdot (V+\varepsilon)+ \frac{t-q+c}{t}\cdot V $ & $V$ \\ 	\hline
			$s_2$ & $V$ & $V$ & $V$ \\ 	\hline
		\end{tabular}
	\end{center}
	\caption{Expected payoff of a real voter in the Devil's Menu with four prices---depending on the slot  applied for (column) and the final status of  the voter's district (column).}
	\label{table_expected_payoffs_realvoter}
\end{table}

\subsection{The weak form of the Devil's Menu}\label{subsec:main_result}

With the four-price Devil's Menu in place, and taking into account the optimal choices of Step 4 discussed above, we can define a simultaneous-move game, which we denote by~$\mathcal{G}$. The player set is $N$, and for each citizen $i$, his strategy set is $S_i=\{s_1,s_2\}$. The payoff matrices are given by Tables~\ref{table_payoffs}~and~\ref{table_expected_payoffs_realvoter}, with the selection of districts determined in accordance with Step 3 of the Devil's Menu. The following observations follow  immediately: First, for each citizen, it is a weakly dominated strategy to abstain and not to apply for either of the slots. We shall assume that no player will play a weakly dominated strategy,  i.e. all citizens will accept to participate in the Devil's Menu. Second, for a real voter, applying for $s_1$ weakly dominates applying for $s_2$. After these assumptions, we obtain our first main result.\footnote{If $q=\bar k$, the result of Theorem \ref{thm:main} is still valid. Nevertheless, it does not yield any budget savings for the adversary with respect to the brute-force attack.} 


\begin{theorem}\label{thm:main}\normalfont{\bf [Weak form of the Devil's Menu]}
Suppose $\delta\geq\frac{q}{\bar{k}}V+2\varepsilon$. Then the strategy profile $\sigma^*$, in which all decoy voters apply for slot~$s_2$ and all real voters apply for slot~$s_1$, is the only Nash equilibrium of game~$\mathcal{G}$.
\end{theorem}
\begin{proof}
Let $\sigma=(\sigma_i)_{i\in N}$ be a strategy profile that constitutes a Nash equilibrium of $\mathcal{G}$. Because all citizens with a real ballot will apply for $s_1$, we can assume w.l.o.g. that, after breaking ties, we have
\begin{equation*}
1\leq \rho_1 \leq \rho_2 \leq \ldots \leq \rho_{\bar{k}}.
\end{equation*}
We next observe that in any equilibrium,  set $O$ is empty and thus $o$ is equal to $0$. Indeed, assume that $o>0$. Then, in any of the districts $k$ that belong to $O$ such that $\rho_k > 1$, any decoy voter $i$ for which $\sigma_i=s_1$ will have an incentive to deviate and apply for slot~$s_2$. This deviation strictly improves expected utility. The reason is as follows: As a consequence of the deviation, either district $N_k$ becomes a selected district, and the expected utility is strictly larger than $\varepsilon$ (since $\delta>\varepsilon$), or district $N_k$ still belongs to $O$, in which case the deviation increases the payment to citizen $i$ by $\varepsilon$.

Since $o=0$ in any Nash equilibrium, we can assume that $t= \bar{k}-c$. We prove next that $\sigma$ cannot be a Nash equilibrium if $\rho>1$. This property can be derived as follows: Suppose  that $\rho>1$. Then, in each district belonging to $T$, there is at least one decoy voter who has applied for slot~$s_1$ in $\sigma$. Consider one of these districts, say $N_k$. The probability that $N_k$ will be selected by fair randomization is equal to $\frac{q-c}{t}$.  Now let $i$ be any decoy voter for which $\sigma_i=s_1$ and who belongs to $N_k$. Then, $i$'s expected payoff is equal to $\frac{q-c}{t}\cdot (V+\varepsilon)+\frac{t-q+c}{t}\cdot\varepsilon$. By contrast, if $i$ deviates from his strategy in $\sigma$ and applies for $s_2$, set $T$ will consist of $t-1$ districts, and set $C$ will consist of $c+1$ districts. In this latter case, district $k$ is selected with probability $1$, and the payoff to decoy voter $i$ is equal to $ \delta$. The latter payoff is strictly higher than $\frac{q-c}{t}\cdot (V+\varepsilon)+\frac{t-q+c}{t}\cdot\varepsilon$, since
\begin{equation*}
\frac{q-c}{t} (V+\varepsilon)+\frac{t-q+c}{t}\varepsilon =\frac{q-c}{\bar{k}-c} V+\varepsilon \leq\frac{q}{\bar{k}} V +\varepsilon<\frac{q}{\bar{k}}V+2\varepsilon\leq\delta.
\end{equation*}
Accordingly, it must be that $\rho=1$. To sum up, we have shown that in any Nash equilibrium of game~$\mathcal{G}$, we have $\rho=1$ and $o=0$. This completes the proof. 
\end{proof}

The rationale for the functioning of the Devil's Menu with four prices as described in Theorem~\ref{thm:main} is clear. By offering different prices---whose  realizations depend on the behavior of the citizenry, both at the district and  the aggregate level---, decoy voters are trapped. These citizens would like to be in slot~$s_1$ and simultaneously ensure that their district is selected by the adversary. Applying for~$s_1$, however, reduces the chances of their district being selected. As a consequence, the Devil's Menu generates downward pressure on $\rho_1,\ldots,\rho_{\bar{k}}$, which results in only real voters applying for slot~$s_1$ in all districts. In particular, only real voters of the selected  districts are paid the high price~$V+\varepsilon$, with decoy voters in these same districts obtaining~$\delta$ instead. In turn, decoy voters in non-selected districts obtain~$2\varepsilon$. Real voters in non-selected districts will not sell their ballots. All in all, the budget that the adversary pays in equilibrium will never be higher than the following bound:
\begin{equation}
\max_{S \subseteq  \{1,\ldots,\bar{k}\}, s=q} \left[ (V+\varepsilon) \cdot \sum_{k\in S} n^R_k+   \delta \cdot \sum_{k\in S} n^D_k + 2\varepsilon   \cdot\sum_{k\notin S}n^D_k  \right]=:\bar{B}.\label{eqn:budget}
\end{equation} 
Of course, the adversary will choose the lowest possible value of~$\delta$ to minimize her costs, so she will set~$\delta$ equal to $\frac{q}{\bar{k}}V+2\varepsilon$.

\subsection{Critical and non-critical assumptions}\label{subsec:assumptions}

In the following, we discuss some critical and less critical assumptions of the Devil's Menu. First, observe that the budget bound~$\bar{B}$ is not only sufficient to buy the real ballots from $q$ districts on the equilibrium path, it would also be sufficient to match expenditures if individual decoy voters deviated from the unique equilibrium. Indeed, Theorem~\ref{thm:main} and the construction of the Devil's Menu imply the following result:

\begin{Cor}\label{cor:sabotage}
	The budget bound $\bar{B}$ is equal to, or higher than, the expenditures faced by the adversary when one decoy voter deviates and applies for slot~$s_1$.
\end{Cor}

This corollary follows directly from the observation that a deviation  by one decoy voter  will turn his district into a non-selected district with probability 1. Accordingly, this voter cannot cause a budget problem for the adversary, nor can he sabotage the functioning of the Devil's Menu with four prices. We further note that to guarantee that the procedure can be run---and that the equilibrium described in Theorem~\ref{thm:main} is unique---, it suffices  for the adversary's budget $B$  to satisfy Inequality~(\ref{eq:budget_condition}) and for this to be common knowledge.

Second, in real voting settings, it is reasonable to expect the number of real voters in any given district to be fixed---and easy to guess for the adversary---, since otherwise the aggregate preferences of the society may not be well represented by the citizens receiving the real ballots. As already discussed, the Devil's Menu proceeds on the assumption that the adversary knows the number of real ballots and the number of citizens in any district. This premise is further discussed in Section~\ref{subsec:noise}.

 
Third, another assumption of our model is that citizens make their applications simultaneously. However, this is not crucial for our results. Indeed, at any point in time where some citizens have already applied for slots as prescribed by Theorem~\ref{thm:main}, the best response of any remaining citizen is to apply for a slot in the same way. That is, the unique equilibrium outcome of Theorem~\ref{thm:main}  also arises as the unique equilibrium of any dynamic version of game~$\mathcal{G}$. Whether dynamic or not, computing the optimal strategy is an easy task for any voter, so the requirements to behave rationally are neither demanding nor unrealistic.

Finally, we note that implementation of any of the above mechanisms can be effected using so-called  {\it Smart Contracts}.\footnote{We refer to \cite{smart}. See also \url{https://en.wikipedia.org/wiki/Smart_contract}, retrieved on 16 November, 2017.} These contracts are computer protocols intended to facilitate, verify, and enforce the exchange between individuals. The Devil's Menu can be coded in any programming language and run on the blockchain. Payments may be carried out in any of the crypto-currencies implementing Smart Contracts. Since the main property of Smart Contracts is that they are self-executing and self-enforcing, the adversary will be committed to the protocol and the payments. This, in turn, justifies the assumptions we made for Theorem~\ref{thm:main} about citizen behavior. Citizens can join the contract at any time before timeout by sending their votes in the form of passwords. Smart Contracts also provide security and partial anonymity, which are paramount in our voting set-up.

\section{Buying Decoy Ballots Cheaper}\label{sec:costless_buying}

In the preceding section we have shown that the adversary can buy all real ballots of $q$ districts, but she cannot avoid paying the price $\delta$ to all decoy voters in those districts. In this section we outline three ways in which this price can be lowered.

\subsection{Sequential vote buying}

 First, the adversary may proceed sequentially and buy only the real ballots of one district at a particular point in time. Specifically, suppose that there are~$q$ points in time and the adversary commits to buy the real ballots of only one district at each date. We obtain the following result:
 
\begin{proposition}\label{prop:sequential}
	Suppose $\delta\geq\frac{1}{\bar{k}-q+1}V+2\varepsilon$. If the adversary uses the Devil's Menu with four prices sequentially and buys the real ballots of one district at each date, the strategy profile $\sigma^*$ in which all decoy voters apply for slot~$s_2$ and all real voters apply for slot~$s_1$ is the only subgame perfect equilibrium. 
\end{proposition}
\begin{proof}
	Proposition~\ref{prop:sequential} can be proved by backward induction. At the last date, there are $\bar{k}-q+1$ districts, of which only one will be selected by the adversary. The corresponding critical value for $\delta$ such that $\sigma^*$ is the only Nash equilibrium at the last date is thus
	\begin{equation*}
		\frac{1}{\bar{k}-q+1}\left(V+\varepsilon\right)+ \varepsilon\leq\frac{1}{\bar{k}-q+1}V+2\varepsilon.
	\end{equation*}
	At all preceding dates, we have a smaller critical value for $\delta$,  which ensures that $\sigma^*$ will prescribe only those actions  at this stage that are compatible with equilibrium behavior, anticipating that $\sigma^*$ will also determine the actions chosen at any future date. The reason is that there are more districts from which one is selected. This completes the proof. 
\end{proof}

\subsection{A strong form of the Devil's Menu: four prices}\label{subsec:strong_form_part_I}

There are at least two alternative versions of the Devil's Menu that require only negligible payments for decoy ballots. Both alternatives---which are manifestations of the so-called \textit{Strong Form of the Devil's Menu}---enable the adversary to use {nearly her} entire budget for real ballots, thereby rendering the prevention of vote buying via decoy ballots completely ineffective. The first alternative follows the same mechanism as the Devil's Menu outlined in Section~\ref{sec:mechanism}, but the price setting is slightly different. Specifically, the price scheme is as follows:

\begin{table}[H]
	\begin{center}
		\def\arraystretch{1.5}
		\begin{tabular}{ c|c|c| } 
			& slot $s_1$ & slot $s_2$ \\  		\hline
			district  selected & $p_1^{se}=V+\varepsilon$ & $p_2^{se}=2\varepsilon$ \\ 	\hline
			district not selected & $p_1^{ns}=\varepsilon$& $p_2^{ns}=2\varepsilon$ \\ 	\hline
		\end{tabular}
	\end{center}
		\caption{The Strong form of the Devil's Menu---depending on the slot  applied for (column) and the final status of the voter's district (row).}
	\label{table_prices_strong}
\end{table}

Note that  by setting $\delta=2\varepsilon$ Table~\ref{table_prices_strong} follows from Table~\ref{table_prices}  and that, moreover, only three different prices are actually offered. The corresponding simultaneous-move game is denoted by $ \widehat{\mathcal{G}}_4$. We obtain the following result:

\begin{proposition}
The strategy profile  $\sigma^*$,  in which all decoy voters apply for slot~$s_2$ and all real voters apply for slot~$s_1$, is a Nash equilibrium of game~$ \widehat{\mathcal{G}}_4$.
\end{proposition}

\begin{proof}
Under this modified Devil's Menu with {four} prices, applying for $s_1$ weakly dominates applying for $s_2$ in the case of real voters. If citizens behave in accordance with $\sigma^*$, the payoff of a decoy voter is $2\varepsilon$, with every district having the same chance {of being} selected. A deviation by a decoy voter $i$ switching from $\sigma_i=\sigma^*=s_2$ to $\sigma_i=s_1$ would then cause the deselection of his district. In such {a} case, citizen $i$ would end up with a payoff of $\varepsilon$. This means that such a deviation is not profitable, which completes the proof. 
\end{proof}
Proposition 2 has important consequences. Since $\varepsilon$ can be chosen arbitrarily small, the budget of the adversary can be used almost entirely for buying real ballots in the equilibrium described in Proposition 2. Total expenditures are then bounded by
 \begin{equation*}
 \max_{S \subseteq  \{1,\ldots,\bar{k}\}, s=q}  \left(V +\varepsilon\right)\cdot \sum_{k\in S} n^R_k +2\varepsilon\cdot \sum_{k\in\{1,\ldots,\bar{k}\}} n^D_k .
 \end{equation*}
{One} possible drawback of this strong form of the Devil's Menu is that game~${\widehat{\mathcal{G}}_4}$ may have other equilibria in which some decoy voters also apply for slot~$s_1$. To eliminate these other equilibria, one has to introduce further refinements {in} the price offering. This constitutes the second part of the strong form of the Devil's Menu.

\subsection{A strong form of the Devil's Menu: {Six} prices}\label{subsec:strong_form_part_II}

An alternative way {of reducing} the cost of buying decoy ballots almost entirely is to enlarge the menu of prices. The price scheme of the second strong form of the Devil's Menu is as follows:
\begin{table}[H]
	\begin{center}
		\def\arraystretch{1.5}
		\begin{tabular}{ c|c|c| } 
			& slot $s_1$ & slot $s_2$ \\  		\hline
			district in set $C$ (selected) & $V+\varepsilon$ & $V-\varepsilon$ \\ 	\hline
			district in set $T$ is selected & $V+\varepsilon$ & $\delta$ \\ 	\hline
			district is not selected & $\varepsilon$& $2\varepsilon$ \\ 	\hline
		\end{tabular}
	\end{center}
		\caption{The {Strong form of the} Devil's Menu with {six} prices---depending on the slot {applied for} (column) and the interim and final status of {the voter's} district (row).}
	\label{table_more_prices}
\end{table}

In this {six}-price setting procedure, $\varepsilon>0$ is assumed to be arbitrarily small. It will also suffice to choose $\delta$ such that $\delta\geq 3 {\varepsilon}$. In addition, note that in Table \ref{table_more_prices} only five different prices are actually offered. The corresponding simultaneous-move game is denoted by ${\widehat{\mathcal{G}}_6}$. We obtain the following result:

\begin{theorem}\label{thm:moreprices}\normalfont{\bf [A strong form of the Devil's Menu]}
	 {Suppose $\delta\geq 3 {\varepsilon}$. Then} the strategy profile $\sigma^*$, in which all decoy voters apply for slot~$s_2$ and all real voters apply for slot~$s_1$, is the only Nash equilibrium of game~${\widehat{\mathcal{G}}_6}$.
\end{theorem}
\begin{proof}
	As in the proof of Theorem 1, abstaining is weakly dominated for all citizens, while applying for slot~$s_2$ is weakly dominated for all real voters. We also observe that $o=0$ must hold again in any equilibrium $\sigma=(\sigma_i)_{i\in N}$. Indeed, if $o>0$, there must exist decoy voters in districts in $O$ who applied for slot~$s_1$. However, these voters can strictly improve their payoff by applying for slot~$s_2$. The reason is as follows: Let $i\in N_k$ be a decoy voter such that $N_k\in T$ and $\sigma_i=s_1$. If $i$ deviated and applied instead for $s_2$, either district $N_k$ would become a selected district and the expected utility strictly larger than $\varepsilon$ (since $V-\varepsilon>\varepsilon$), or district $N_k$ would still belong to $O$, in which case the deviation would increase the payment to citizen $i$ (since $\delta>2\varepsilon$). Hence, we obtain that $o=0$, and it follows that $t=\bar{k}-c$.
	
	Next we show that $\rho>1$ cannot occur in equilibrium either. Suppose  that $\rho>1$ for $\sigma$. 
	Then, in each district that belongs to $T$, there must be at least one decoy voter $i$ such that $\sigma_i=s_1$. Consider one of these districts, say $N_k$, and a decoy voter $i$ in $N_k$ who has applied for slot~$s_1$. The probability that $N_k$ will be selected by fair randomization is equal to $\frac{q-c}{t}$. Then $i$'s expected payoff is equal to $\frac{q-c}{t}\cdot (V+\varepsilon)+\frac{t-q+c}{t}\cdot\varepsilon$. By contrast, if $i$ deviates from his strategy in $\sigma$ and applies for $s_2$, set $T$ will consist of $t-1$ districts and set $C$ will consist of $c+1$ districts. In this latter case, district $N_k$ will be chosen with probability $1$, and the payoff to decoy voter $i$ will be equal to $V-\varepsilon$. We claim that the latter payoff is strictly larger than $\frac{q-c}{t}\cdot (V+\varepsilon)+\frac{t-q+c}{t}\cdot\varepsilon$ if $\varepsilon$ is sufficiently small, {so} $i$ is better off if he applies for slot~$s_2$. Indeed, using $t=\bar{k}-c$ and $q<\bar{k}$, and taking $\varepsilon>0$ arbitrarily low, we have
	\begin{equation*}
	\frac{q-c}{t} (V+\varepsilon)+\frac{t-q+c}{t}\varepsilon =\frac{q-c}{\bar{k}-c} V+\varepsilon \leq\frac{q}{\bar{k}} V +\varepsilon<V-\varepsilon.
	\end{equation*}
	Finally, $\rho=1$ is an equilibrium if no decoy voter wants to apply for slot~$s_1$. In the unique equilibrium featuring $\rho=1$, the expected payoff for such a voter is $2\varepsilon$. Applying for slot~$s_2$ would result in $2\varepsilon$, since the district to which this citizen belongs would {no longer} be selected. Hence, the deviation is not profitable, which completes the proof. 
\end{proof}

The Devil's Menu with {six} prices efficiently eliminates almost all expenditures on decoy ballots {and again generates more power for buying} real ballots. Indeed, by taking $\delta=3\varepsilon$, total expenditures are now bounded by
\begin{equation*}
\max_{S \subseteq  \{1,\ldots,\bar{k}\}, s=q}  \left(V +\varepsilon\right)\cdot \sum_{k\in S} n^R_k +3\varepsilon\cdot \sum_{k\in\{1,\ldots,\bar{k}\}} n^D_k .
\end{equation*} 
It is however important to point out that this second strong form of the Devil's Menu rests {decisively} on the assumption that the adversary can credibly offer two different prices for selected districts, at least for slot~$s_2$. The two prices discriminate among districts based on their interim status. On the one hand, the district may belong to $C$---and hence its interim status is immediately selected. On the other hand, the district may belong to $T$ and be selected only after it has been chosen by fair randomization.\footnote{It may happen that set $T$ contains one district only, in which case such district will be selected with certainty. Citizens of this district are thus offered prices than citizens of all other selected districts.} This property can be built into the algorithms executed by Smart Contracts, but it may be less accepted by citizens.

\section{Extensions and Applications}\label{sec:extensions}

In this section we do three things. First, we analyze the role of the adversary's commitment power. Second, we investigate some extensions of the baseline framework set out in Section~\ref{sec:mechanism}. Third, we discuss applications of the Devil's Menu to the standard ``Lemons Problem.''

\subsection{A simple Devil's Menu with strong commitment}\label{subsec:strong_commitment}

For all the Devil's Menu variants we have analyzed in the previous sections, it was essential that the adversary could commit to the use of the protocol underlying the procedure. Stronger assumptions on the commitment power of the adversary would open up other forms of the Devil's Menu. A very simple procedure would be sufficient if the adversary could commit to not buying any ballot {if} the number of applicants across slots {did not satisfy} a certain criterion. This can be illustrated for a situation in which there are no districts or, equivalently, any information regarding them is disregarded by the adversary. Recall that $n^R$ denotes the total number of real voters. {Now}, suppose that the adversary wants to buy $l$ real ballots, with $l\leq n^R$. For that purpose, {she} considers the following procedure: 
\begin{enumerate}
	\item Two slots, $s_1$ and $s_2$, are announced.
	
	\item Each citizen applies for one of the two slots.
	
	\item With citizens being able to decide whether to sell their vote or not at the prevailing price and the adversary being obliged to accept the transaction at the request of the citizens, payoffs are realized according to the following rule:
	\begin{itemize}
		\item If more than $n^R$ citizens apply for slot~$s_1$, then all of them are offered $0$ in exchange for their ballot. In turn, all citizens who applied for slot~$s_2$ {are} offered $\varepsilon$, also in exchange for their ballot.
		
		\item If $n^R$ voters {at most} applied for slot~$s_1$, then $l$ of them are chosen  randomly. {They} then receive $V+\varepsilon$ in exchange for their ballot. In turn, all citizens who applied for slot~$s_2$ {are} offered $\varepsilon$, also in exchange for their ballot.
	\end{itemize}
\end{enumerate}

It is easy to see that in the simultaneous-move game induced by the mechanism described above, the strategy profile where all the real citizens apply for the slot~$s_1$ and all the decoy voters apply for the slot~$s_2$ is the only Nash equilibrium, provided that real voters do not play weakly dominated strategies. Moreover, the expenditures used to buy decoy ballots are negligible. Note that the assumption on the commitment power of the adversary is crucial for the functioning of the procedure, as citizens must believe that if the adversary ends up with zero ballots, she will not renege on her promise not to try and buy some.

One important drawback of this mechanism is that, unlike the Devil's Menu with {four} or more prices---see e.g. Corollary~\ref{cor:sabotage}---, any single decoy voter has the power to {easily} sabotage the mechanism by applying for slot~$s_1$. In such cases, repeating the mechanism will not help. The reason is the following: If there is a chance that the mechanism will be repeated, incentives will be created  for any decoy voter to deviate in the first round and apply for slot~$s_1$. 

\subsection{Noise}\label{subsec:noise}

To counteract the power of the adversary when using the Devil's Menu, an election designer might want to introduce noise. One possibility would be {random choice of the number} of real voters across districts.\footnote{We {emphasize} that the functioning of the Devil's Menu is not affected if there is uncertainty about the number of decoy ballots, provided that no citizen can hold more than one ballot.} {On} other occasions, the adversary might simply not know the precise number of real voters in each district, particularly if there is no natural partition of the entire population into small subgroups upon which the procedure can be built. For instance, there may be {only} one district. In either of these situations, the Devil's Menu could still be applied, based on the expected number of real voters in each district. In doing so, however, some decoy voters may be paid the higher price~$V+\varepsilon$. This would result in overspending by the adversary.


\subsection{General valuation distribution}\label{subsec:general_valuation}

In our analysis thus far, we {have} assumed that all real voters had a common valuation $V$ for their ballots. In a real environment, however, voter valuations will be distributed within some range. In this case, one possibility {of buying} all the real ballots of $q$ districts would be to have $V$ denote the maximum possible voter valuation and then run the Devil's Menu accordingly. This approach could nonetheless be very inefficient, since some real voters may be offered prices that are well above their willingness to sell. Alternatively, the adversary could have $V$ denote a price threshold ensuring that a {large} majority of the population will have their valuations below it with very high probability. If the adversary {ran} the Devil's Menu with this parameter, real voters with valuations below $V$ would participate and those above would abstain. Given sufficiently accurate information on the distribution, the Devil's Menu might still be reasonably efficient. 

Another feature of our analysis that was missing in the previous sections is that in both elections or referenda ballots are to be cast in favor of some alternative. This implies that a ballot is not a standard consumption good. In particular, voters may obtain the same utility regardless of who is eventually casting their ballot, as long as it is done in favor of their preferred alternative. Accordingly, if the adversary herself {had} some preference regarding the alternatives at hand---and hence would not simply act as a reseller of ballots---and this preference was common knowledge, the real voters who had the same preference as the adversary could be considered decoy voters. The reason is that they would anticipate that the adversary would use their vote as they intended, so their willingness to sell would drop from $V$ to $0$.\footnote{With only two alternatives, $V$ can be interpreted as the differential utility citizens with different preferences than the adversary would obtain from casting a vote against {their} own preferred alternative.} As discussed above, if the adversary learned the number of real voters who share her preferences, she could still run the Devil's Menu very efficiently by taking into account {the fact} that these citizens {are behaving} as decoy voters. 

\subsection{The Devil's Menu and the ``Lemons Problem''}\label{subsec:lemons}

The Devil's Menu offers further applications {in connection with} the so-called ``Lemons  Problem.''\footnote{The classical {``Lemons Problem''} is outlined in \cite{akerlof1970market}. We refer to \cite{riley2001silver} for an early survey of possible solutions to lemon  problems and to \cite{kim2012endogenous} for partial solutions using specific market organizations. For mechanism design approaches that could be used, see~\cite{borgers2015introduction}.} We outline {one} of them. Suppose that an agent---the \textit{buyer}---wants to buy \textit{one} used car of high quality, and that she can buy it from several individuals---the \textit{sellers}---each {of whom have} a used car. The buyer knows the number of good-quality and bad-quality cars among the sellers, yet she does not know which particular sellers have them. In this situation, the buyer could run the Devil's Menu {in} Section~\ref{subsec:strong_commitment}, offer two slots{---$s_1$ and $s_2$---}to all sellers, and then select one seller from those who {have} applied for slot $s_1$ by fair randomization, from which the buyer would get the car at the prevailing price. This strategy for an adversary could be particularly appealing in the Internet, in which case the purchase of the car could be set up similarly as in online auctions.

\section{Conclusion}\label{sec:conclusion}

We have shown that electronic voting systems based on decoy ballots may be vulnerable to sophisticated attacks, particularly when the population with {a} right to vote is divided into subgroups and the adversary exhibits some degree of commitment power. In such cases, the adversary may significantly reduce her expenditures on decoy ballots and thus use her budget mainly to buy real ballots. This would hence render decoy ballots an ineffective tool {for fighting} vote buying by trying to blow up the adversary's budget. 

{While} we have limited ourselves to a simple scenario, {in our analysis we have also} touched upon some more general scenarios, the comprehensive understanding of which merits analysis in {each particular case}. Examples are the existence of imperfect information about the number of real ballots and applications of the Devil's Menu to {the} ``Lemons Problem.'' Numerous further scenarios also lie ahead of us. {We refer to three of them.} First, citizens may have more than one ballot or, similarly, they could collude and adopt a unified strategy. Second, there could be two adversaries trying to buy votes, possibly with conflicting interests. Third, as a further application to the ``Lemons Problem'', the Devil's Menu could be tailored to situations where an agent---the adversary---wants to buy a valuable secret from a group of individuals, say, a firm or a bureaucratic unit, knowing that only one of the individuals has it. Examples could be a password or the identity of the person who has important knowledge. All these instances are left to future research.




\bibliographystyle{apalike}

\bibliography{decoy}

\end{document}